\def\Title{Reconstructing Bohr's Reply to EPR 
in Algebraic Quantum Theory}
\def\AuthorA{Masanao Ozawa\thanks{email: ozawa@is.nagoya-u.ac.jp}}
\def\AffiliationA{\small\it Graduate School of Information Science, Nagoya University,
Nagoya 464-8601,\\  \small\it Japan}
\def\AuthorB{Yuichiro Kitajima\thanks{email: kitajima.yuichirou@nihon-u.ac.jp}}
\def\AffiliationB{\small\it College of Industrial Technology, Nihon University,
\small\it  2-11-1 Shin-ei, Narashino,\\ \small\it Chiba 275-8576, Japan}
\def\Abstract{
Halvorson and Clifton have given a mathematical reconstruction
of Bohr's reply to Einstein, Podolsky and Rosen (EPR), and argued that 
this reply is dictated by the two requirements of classicality and objectivity 
for the description of experimental data, 
by proving consistency between their objectivity requirement and 
a contextualized version of the EPR reality criterion which had been introduced 
by Howard in his earlier analysis of Bohr's reply.
In the present paper, we generalize the above consistency theorem, 
with a rather elementary proof, to a general formulation of EPR states 
applicable to both non-relativistic 
quantum mechanics and algebraic quantum field theory; and we clarify
the elements of reality in EPR states in terms of Bohr's requirements of classicality
and objectivity, in a general formulation of algebraic quantum theory.
}
\def\Keywords{Bohr, EPR, algebraic quantum theory, modal interpretation}
\documentclass{article}
\evensidemargin=\oddsidemargin
\usepackage{amsthm}
\usepackage{amsmath,amssymb}
  \newtheorem{theorem}{Theorem} 
  \newtheorem{proposition}{Proposition} 
  \newcommand{\beq}{\begin{equation}}
  \newcommand{\eeq}{\end{equation}}
  \newcommand{\beql}[1]{\begin{equation}\label{eq:#1}}

  \newcommand{\beqa}{\begin{eqnarray}}
  \newcommand{\eeqa}{\end{eqnarray}}
  \newcommand{\beqas}{\begin{eqnarray*}}
  \newcommand{\eeqas}{\end{eqnarray*}}
  
  \newcommand{\C}{{\bf C}}
  \newcommand{\R}{{\bf R}}
  \newcommand{\al}{\alpha}
  
  \newcommand{\ch}{\chi}

  \newcommand{\om}{\omega}

  \newcommand{\rh}{\rho}

  \newcommand{\vp}{\varphi}

  \newcommand{\Om}{\Omega}

  \newcommand{\Eq}[1]{Eq. (\ref{eq:#1})}

\newcommand{\bracket}[1]{\langle#1\rangle}
  
  \newcommand{\bB}{{\bf B}}

  \newcommand{\bS}{{\bf S}}

  \newcommand{\cA}{{\cal A}}
  \newcommand{\cB}{{\cal B}}

  \newcommand{\cF}{{\cal F}}
  
  \newcommand{\cH}{{\cal H}}

  \newcommand{\cO}{{\cal O}}

  \newcommand{\cS}{{\cal S}}
  
  \newcommand{\cU}{{\cal U}}

  \newcommand{\hA}{{\hat A}}
  \newcommand{\hB}{{\hat B}}

  \newcommand{\hP}{{\hat P}}
  \newcommand{\hQ}{{\hat Q}}

\newcommand{\OO}{\cO}
\newcommand{\CC}{C^*}

\newcommand{\AP}{{\rm AP}}
\newcommand{\mA}{\mathfrak{A}}
\newcommand{\mB}{\mathfrak{B}}
\newcommand{\mN}{\mathfrak{N}}
\renewcommand{\cA}{\mA}
\renewcommand{\cB}{\mB}
\renewcommand{\bB}{\mathbb{B}}
\newcommand{\vpi}{\varpi}
\renewcommand{\rho}{\varphi}
\title{\bf \Title}
\author{\sc \AuthorA \\ 
\AffiliationA\\ \\
\sc \AuthorB \\ 
 \AffiliationB
}
\date{}
\begin{document}
% TITLE AND ABSTRACT:
\maketitle
\begin{abstract}
\Abstract
%PACS: \PACSnumbers
\end{abstract}

%%%%% TEXT %%%%%% 

%%%%% TEXT %%%%%% 

\section{Introduction}
Einstein, Podolsky and Rosen (EPR) \cite{EPR35} discussed a system consisting of
two particles. They have interacted initially and then moved apart so that the positions, 
and the momenta,  of the two particles are strictly correlated.
It follows that if one were to measure the position of the first particle, one could predict with certainty the outcome of a position measurement on the second particle; and similarly for a momentum measurement.
EPR proposed a ``criterion of reality'': 
\begin{quote}
If, without in any way disturbing a system, we can predict with certainty (i.e., with probability equal to unity) the value of a physical quantity, then there exists an element of physical reality corresponding to this physical quantity.
\cite[p.777]{EPR35}
\end{quote}
In accordance with this criterion, the position and the momentum of the second particle have simultaneous reality since the measurement on the first particle has not disturbed the second particle. On the other hand, position and momentum cannot 
have simultaneous reality in any states in  
quantum mechanics. Therefore, EPR regarded quantum-mechanical description as incomplete. 

In reply to EPR,
Bohr pronounced that a measurement on the first particle influences the condition which defines elements of reality for the second particle, 
so that he rejected EPR's conclusion.
\begin{quote}
[T]here is in a case like that just considered no question of a mechanical disturbance of the system under investigation during the last critical stage of the measuring procedure. But even at this stage there is essentially the question of \textit{an influence on the very conditions which define the possible types of predictions regarding the future behavior of the system}. Since these conditions constitute an inherent element of the description of any phenomenon to which the term ``physical reality'' can be properly attached, we see that the argumentation of the mentioned authors [EPR] does not justify their conclusion that quantum-mechanical description is essentially incomplete. \cite[p.700]{Boh35}
\end{quote}

Unfortunately, Bohr's view has prevailed under the Copenhagen interpretation
with several improper restatements, and so there have been only a few serious attempts 
to reconstruct his reply in a rigorous analysis.
Here, we shall discuss Howard's original contribution \cite{How79,How94}
to the reconstruction of Bohr's reply to EPR and its mathematical reformulation due to 
Halvorson and Clifton \cite{HC02}.

Howard \cite{How94} reconstructed Bohr's reply to EPR 
in terms of Bohr's notion of  ``classical description'' and ``a contextualized
version of the EPR reality criterion.''
Bohr stated his view on the classicality requirement of description 
of experimental data as follows.
\begin{quote}
[I]t is decisive to recognize that, \textit{however far the phenomena transcend the scope of classical physical explanation, the account of all evidence must be expressed in classical terms}. The argument is simply that by the word ``experiment'' we refer to a situation where we can tell others what we have done and what we have learned and that, therefore, the account of the experimental arrangement and of the results of the  observations must be expressed in unambiguous language with suitable application of the terminology of classical physics. 
\cite[p.209]{Boh49}
\end{quote}
To explicate this, Howard \cite{How79,How94} considered 
a ``measurement context'' $(\psi,R)$, 
consisting of  a state vector $\psi$ and an observable $R$,
in which the observable $R$ is to be measured in the state $\psi$.
He asked which observables should be considered to possess their values
as elements of reality,
and claimed that Bohr's notion of classical description is best explicated 
via the notion of  ``appropriate mixture'' for the measurement context 
$(\psi,R)$ under consideration
 \cite[pp.190--199]{Bub97}.
\begin{quote}
My claim about the nature of a classical description is that Bohr did not mean simply the application of classical physics --- the physics of Newton, Maxwell, Boltzmann, and Einstein --- in some combination appropriate to the occasion. I will argue instead that a classical description, in the sense of 
``classical'' relevant to Bohr's concerns, is a description in terms of what physicists call ``mixtures'' (as opposed to what are termed ``pure cases''), a formal device that permits us to proceed as if the systems being described were in well-defined, if unknown, intrinsic states, at least with respect to those properties requiring a classical description. The device of mixtures also permits one to give a classical, ignorance interpretation to any statistics that one encounters. 
Which mixture to employ in a given classical description will depend upon the kind of measurement being performed, the ``appropriate mixture'' being one constructed out of simultaneous eigenstates of all the observables measurable in a given experimental context. 
\cite[p.203--204]{How94} 
\end{quote}
Halvorson and Clifton \cite[p.13]{HC02}
showed that Howard's notion of ``appropriate mixture''
can be formulated generally in algebraic quantum theory as a beable
subalgebra \cite{HC99} for a given measurement context $(\psi,R)$.

However, this notion is not sufficient to select one appropriate beable
subalgebra from other candidates. 
For this purpose, 
Howard \cite{How79,How94} introduced 
 ``a contextualized version of the EPR reality criterion.'' 
\begin{quote} 
Once the experimental context is stipulated, which amounts to the
specification of the candidates for real status, our decision as to
which particular properties to consider as real will turn on the
question of predictability with certainty.  \cite[p.256]{How79}
\end{quote}
Halvorson and Clifton \cite[pp.14--15]{HC02}
turned down Howard's contextualized EPR reality criterion\footnote{
The reason for their rejection of Howard's contextualized version of 
the EPR reality criterion is obscure.
They argued that the EPR reality criterion is best construed 
as a version of
`inference to the best explanation' following Redhead \cite[p.72]{Red87}
and that we cannot expect Bohr to be persuaded by such inference to the
best explanation since he is not a classical scientific realist \cite[p.11]{HC02}.
This is plausible, and yet it is not clear whether Bohr would contend
against the contextualized version either.}. 
For their interpretation of Bohr's notion of objectivity, they introduced the requirement
of invariance under the `relevant' symmetries 
for the set of observables corresponding 
to elements of reality.
Then they proved the equivalence between Howard's contextualized EPR
reality criterion and the above invariance principle
 for the case of Bohm's simplified spin version of the EPR experiment
 \cite[Theorem 1]{HC02}.
For the EPR position-momentum case 
they proved the consistency between those two requirements
\cite[Theorem 2]{HC02}.

In the present paper, we follow an abstract definition of the EPR state
for a pair of observables due to Werner \cite{Wer99};
and we first show, in Section 2, that such a state exists 
for incommensurable pairs of observables 
in both nonrelativistic quantum mechanics 
and algebraic quantum field theory (Theorem \ref{theorem1} of Section 2). 
Then it is shown that ``contextualized version of the EPR reality criterion'' is consistent with the requirements of Halvorson and Clifton for the case of Bohm's simplified spin version of the EPR experiment, the EPR position-momentum case, and the case of algebraic quantum field theory (Theorems \ref{theorem3},  \ref{theorem5}, and \ref{theorem9}).

\section{EPR states}

In this section, we provide an abstract definition of  EPR states
and show that such a state exists in both nonrelativistic quantum mechanics 
and algebraic quantum field theory (Theorem \ref{theorem1}).

Let $\cA$ be a unital C*-algebra.
A {\em commuting pair}  in $\cA$ is a pair $(A,B)$ of 
commuting self-adjoint elements $A, B\in \cA_{sa}$.
A state $\vp$ of $\cA$ is called an {\em EPR state} 
\cite{AV00,Wer99}
for a commuting
pair $(A,B)$ if $\vp((A-B)^2)=0$.
For any commuting pair $(A,B)$,  the {\em joint probability distribution} 
$\mu^{A,B}_{\vp}$ of $A,B$ in $\vp$ is defined uniquely
to be a probability
measure $\mu$ on $\R^2$ such that
\beql{JPD}
\vp(f(A,B))=\int_{\R^{2}}f(x,y)\,d\mu^{A,B}_{\vp}(x,y)
\eeq
for any polynomial $f(A,B)$.
Then, $\vp$ is an EPR state if and only if $\mu^{A,B}_{\vp}$
is concentrated on the diagonal, i.e, 
$\mu^{A,B}_{\vp}(\{(x,x)|\ x\in\R\})=1$.
Thus, simultaneous measurements of $A$ and $B$ always give
concordant results, and each one of the outcomes would predict with certainty 
the other.
In Ref.~\cite{06QPC} one of the present authors introduced a general notion 
of  ``quantum perfect correlation'', which can be applicable to any pair of 
not necessarily bounded observables $A,B$ that do not necessarily 
commute, and showed that this notion can be equivalently characterized
by several mathematical conditions, including the condition that they 
have the joint probability distribution concentrated on the diagonal;
a quantum logical reconstruction of this notion has been given in 
Ref.~\cite{11QRM}.  According to the above, 
an EPR state $\vp$ for $(A,B)$ realizes a ``quantum perfect correlation'' 
between commuting $A$ and $B$ in the state $\vp$.

Let $\cA_1$ and $\cA_2$ be mutually commuting C*-subalgebra
of $\cA$.
Commuting pairs $(A_1,A_2)$ and $(B_1,B_2)$ in $\cA$ are called
{\em incommensurable pairs} from $(\cA_1,\cA_2)$ for a state $\vp$ of $\cA$
if $A_j,B_j\in\cA_j$ and there are integers $n_j,m_j$ such that 
$\vp(|[A_j^{n_j},B_j^{m_j}]|^2)\not=0$ for $j=1,2$, 
where $[A, B]$ denotes the commutator of $A$ and $B$, i.e.,
$[A,B]=AB-BA$.
Here, we note that the incommensurability here is related to the 
state-dependent notion of non-commutativity.  
We say that observables $A$ and $B$ {\em commute in a state $\vp$} 
if there exists a joint probability distribution $\mu^{A,B}_{\vp}$ of $A, B$ 
in $\vp$ satisfying \Eq{JPD} for any polynomial $f(A,B)$.
Then, it can be seen that $A$ and $B$ commute in $\vp$ if and only
if $\vp(|[A^n,B^m]|^2)=0$ for any integers $n,m$ (See, for instance, Ref.~\cite{06QPC}).
Thus, commuting pairs $(A_1,A_2)$ and $(B_1,B_2)$ from $(\cA_1,\cA_2)$ 
are incommensurable for $\vp$ if and only if $A_j$ and $B_j$ 
do not commute in $\vp$ for $j=1,2$.

\sloppy
Let $\vp$ be an EPR state for two incommensurable pairs 
$(A_1,A_2)$ and $(B_1,B_2)$. Then
$\mu^{A_1,A_2}_{\vp}(\{(x,x)|\ x\in\R\})=1$
and $\mu^{B_1,B_2}_{\vp}(\{(x,x)|\ x\in\R\})=1$.
It follows that if we were to measure $A_1$ in $\rho$, we could predict with certainty the outcome of $A_2$; and if we were to measure $B_1$ in $\rho$, we could predict with certainty the outcome of $B_2$.

Let us consider Bohm's spin version of the EPR experiment. Suppose that we have prepared an ensemble of spin-1/2 particles in the singlet state
\[ \Psi=\frac{1}{\sqrt{2}} \left( \begin{pmatrix} 1 \\ 0 \end{pmatrix} \otimes \begin{pmatrix} 0 \\ 1 \end{pmatrix} - \begin{pmatrix} 0 \\ 1 \end{pmatrix} \otimes \begin{pmatrix} 1 \\ 0 \end{pmatrix} \right), \]
and let $\vp$ be the vector state of $\mathbb{B}(\mathcal{H}_2) \otimes \mathbb{B}(\mathcal{H}_2)$ induced by $\Psi$, where $\mathbb{B}(\mathcal{H}_2)$ is the space of operators on 
the 2-dimensional Hilbert space $\mathcal{H}_2=\C^2$. 
Then, 
$\vp$ is an EPR state for incommensurable pairs 
\[ \left( \begin{pmatrix} 1 & 0 \\ 0 & 0 \end{pmatrix} \otimes I, I \otimes \begin{pmatrix} 0 & 0 \\ 0 & 1 \end{pmatrix} \right) \]
and
\[ \left( \frac{1}{2} \begin{pmatrix} 1 & 1 \\ 1 & 1 \end{pmatrix} \otimes I, I \otimes \frac{1}{2} \begin{pmatrix} 1 & -1 \\ -1 & 1 \end{pmatrix} \right)\]
from $(\bB(\cH_2)\otimes \C I,\C I\otimes \bB(\cH_2))$.

In the rest of this section, we shall consider algebraic quantum field theory. 
In algebraic quantum field theory, each bounded open region $\mathcal{O}$ in Minkowski space is associated with a von Neumann algebra $\mathfrak{N}(\mathcal{O})$. Such a von Neumann algebra is called a {\em local algebra}. 
We say that bounded open regions $\mathcal{O}_1$ and $\mathcal{O}_2$ are 
{\em strictly space-like separated} 
if there is a neighborhood $\mathcal{V}$ of the origin of Minkowski space 
such that $\mathcal{O}_1 + \mathcal{V}$ and $\mathcal{O}_2$ are space-like separated.

In the present paper, we make the following assumptions. 
{\em For any bounded open region $\mathcal{O}$ in Minkowski space, $\mathfrak{N}(\mathcal{O})$ is not abelian. 
If $\mathcal{O}_1$ and $\mathcal{O}_2$ are space-like separated, 
then $[A_1,A_2]=0$ 
for any $A_1 \in \mathfrak{N}(\mathcal{O}_1)$ and $A_2 \in \mathfrak{N}(\mathcal{O}_2)$. If $\mathcal{O}_1$ and $\mathcal{O}_2$ are strictly space-like separated, then $A_1A_2 \neq 0$ for any nonzero operators $A_1 \in \mathfrak{N}(\mathcal{O}_1)$ and $A_2 \in \mathfrak{N}(\mathcal{O}_2)$} 
(cf. \cite[Theorem 1.12.3]{Bau95}).

The following theorem shows that there is an EPR state of $\mathfrak{N}(\mathcal{O}_1) \vee \mathfrak{N}(\mathcal{O}_2)$ for incommensurable pairs from $(\mN(\cO_1),\mN(\cO_2))$ 
if $\mathcal{O}_1$ and $\mathcal{O}_2$ are strictly space-like separated, where $\mathfrak{N}(\mathcal{O}_1) \vee \mathfrak{N}(\mathcal{O}_2)$ is the von Neumann algebra generated by $\mathfrak{N}(\mathcal{O}_1)$ and $\mathfrak{N}(\mathcal{O}_2)$.

\begin{theorem}
\label{theorem1}
Let $\mathfrak{N}_1$ and $\mathfrak{N}_2$ be mutually commuting 
non-abelian von Neumann algebras on a Hilbert space $\mathcal{H}$ 
such that $A_1A_2 \neq 0$ 
for any nonzero operators $A_1 \in \mathfrak{N}_1$ and $A_2 \in \mathfrak{N}_2$. 
Then, there exists 
a vector state $\vp$ of $\mathfrak{N}_1 \vee \mathfrak{N}_2$ 
which is an EPR state for incommensurable
pairs $(E_1,E_2)$ and $(F_1,F_2)$ of projections from 
$(\mathfrak{N}_1, \mathfrak{N}_2)$.
\end{theorem}

\begin{proof}
Since  $\mathfrak{N}_1$ is not abelian, there exists projections $R$ and $S$ 
such that $[R,S] \neq 0$. 
Define $T\in \mathfrak{N}_1$ by $T= (I-R)SR$. 
Then $T \neq 0$. 
Let $T=V \vert T \vert$ be the polar decomposition of $T$. 
It follows from the construction of $T$
that $VV^*$ is orthogonal to $V^*V$ and $V^2=0$ 
(cf. \cite [Lemma]{Lan87}). 
Similarly, there exists a partial isometry $W \in \mathfrak{N}_2$ 
such that $WW^*$ is orthogonal to $W^*W$ and $W^2=0$. Let
$E_1 = VV^*$, 
$F_1 = \frac{1}{2}(VV^*+V^*V+V+V^*)$, 
$E_2 = WW^*$,  and 
$F_2 = \frac{1}{2}(WW^*+W^*W+W+W^*)$.
Then $F_1$ and $F_2$ are projections in $\mathfrak{N}_1$ and $\mathfrak{N}_2$, respectively.
Since $E_1 \in \mathfrak{N}_1$ and $E_2 \in \mathfrak{N}_2$, 
we have $E_1E_2 \neq 0$ by hypothesis. Let $\Psi_1$ be a unit vector such that 
$\Psi_1 \in E_1E_2\mathcal{H}$.
Since $VW(V^*W^*)\Psi_1=\Psi_1$, we have $V^*W^*\Psi_1 \neq 0$. Define $\Psi \in\cH$ by
$ \Psi = \frac{1}{\sqrt{2}}(\Psi_1+V^*W^*\Psi_1)$. 
Then $\Psi$ is a unit vector. Let $\vp$ be the vector state of 
$\mathfrak{N}_1 \vee \mathfrak{N}_2$ induced by  $\Psi$. Then
\[ \vp(E_1)=\vp(E_2)=\vp(E_1E_2)=\vp(F_1)=\vp(F_2)=\vp(F_1F_2)=\frac{1}{2}. \]
Since $E_1F_1 \Psi =\frac{1}{2\sqrt{2}}(\Psi_1+W^*\Psi_1)$, 
$F_1E_1\Psi=\frac{1}{2\sqrt{2}}\Psi_1$ and $W^* \Psi_1 \neq 0$, 
we have $\vp(|[E_1,F_1]|^2) \neq 0$. Similarly $\vp(|[E_2,F_2]|^2) \neq 0$.
Therefore, the vector state $\vp$ of $\mathfrak{N}_1 \vee \mathfrak{N}_2$ is 
an EPR state for the incommensurable pairs $(E_1,E_2)$ and $(F_1,F_2)$ from
$(\mathfrak{N}_1, \mathfrak{N}_2)$. 
\end{proof}

\section{Bohr's requirements of classicality and objectivity}

Halvorson and Clifton followed Howard's interpretation of Bohr's notion of 
``classical description,'' and they reformulated this notion in algebraic quantum
theory as follows \cite[p.13]{HC02}.
Let $\cB$ be a unital C*-subalgebra of a unital C*-algebra $\cA$
with a state $\vp$.
A state $\omega$ of $\mathfrak{B}$ is said to be dispersion-free 
if $\omega(A^*A)= \vert \omega(A) \vert^2$ for any $A \in \mathfrak{B}$.
We say that $\vp$ is {\em classical} on $\cB$ or 
that $\cB$ is {\em beable} for $\vp$
 if  there is  a probability
measure $\mu$ on the space $\cS_{DF}(\cB)$ of dispersion-free 
states of $\cB$ satisfying
\beql{101025}
\vp(A)=\int_{\cS_{DF}(\cB)}\om(A)d\mu(\om)
\eeq
for every $A\in\cB$.

A {\em measurement context} is defined to be a pair $(\vp,A)$ of 
 a state $\vp$ of $\cA$  and
a self-adjoint element $A\in\cA$.
If $A\in\mA$ is an observable being measured in state $\rho$, 
then $A$ must be contained in a beable subalgebra $\mB$ 
determined by the measurement context $(\vp,A)$,
since the outcome of the $A$ measurement must be described classically
by the classical probability distribution $\mu$ of possible values $\om(A)$
with $\om\in\cS_{DF}(\cB)$.
Then, the subalgebra $\cB$ contains only observables for which we can
speak of their values in terms of classical language.

As discussed in Section 1, we need some proposal to uniquely determine
an appropriate subalgebra ${\cB}$.
Howard's strategy to single out such a subalgebra is to collect all the 
observables whose value can be predictable from the value of $A$ 
without any further disturbance of the system.
This can be done if $\vp$ is an EPR state for $(A,B)$.
Thus, Howard's strategy is to collect all observables of the form $f(B)$ such that
$\vp$ is an EPR state for $(A,B)$.
This strategy was successful to reconstruct Bohr's reply 
in Bohm's modified spin version of the EPR argument. 
Halvorson and Clifton called the above doctrine a {\em contextualized 
version of the EPR reality criterion}, since that assigns an element
of reality to an observable the value of which can be predicted without
disturbing the system, provided it can be described by a classical language together with
the value of measured observable $A$.

Halvorson and Clifton \cite[pp.14--15]{HC02}
proposed an alternative approach to single out a suitable set of 
observables to possess their values in the given measurement context.
According to them,
a feature of a system is ``objective'' by 
being ``invariant under the `relevant' group of symmetries'' 
\cite[p.11]{HC02}, as follows.
Let $\mathfrak{B}$ be a unital C*-subalgebra of a unital C*-algebra 
$\mathfrak{A}$ with a self-adjoint element $A$ and a state $\vp$.
We say that $\mathfrak{B}$ is {\em definable in 
the measurement context $(\rho,A)$} 
if $U^*\mathfrak{B}U=\mathfrak{B}$ for any unitary $U\in\cA$
such that $[A,U]=0$ and $\vp_U=\vp$, where $\vp_U$ is defined
by $\vp_{U}(X)=\vp(U^*XU)$ for every $X\in\cA$.
Here, we have used a notion of implicit definability, in the
sense that $\mathfrak{B}$ can be considered to be implicitly defined
by $\vp$ and $A$ if the membership of $\mathfrak{B}$ is not affected
by any automorphism of $\cA$ that leaves $\vp$ and $A$ invariant;
this notion is widely used in foundational studies, for example, by
Malament \cite[p.297]{Mal77}.

Recall from Section 1 that Halvorson and Clifton  proved consistency 
between the above two approaches to single out an appropriate subset of
observables ---
Howard's contextualized version of the EPR reality criterion and
the invariance under the relevant group of symmetry ---
for the position-momentum case.
Now, we shall prove that the above two approaches are consistent even
in a general formulation of algebraic quantum theory.
To formulate the statement, we say that a unital C*-subalgebra $\cB$ of $\cA$ 
is a  {\em beable subalgebra} for
$(\vp,A)$ if it satisfies following conditions:
\begin{enumerate}
\item[(i)]  (Beable) $\vp$ is classical on $\cB$.
\item[(ii)] ($A$-Priv) $A\in\cB$.
\item[(iii)] (Def)  $\mathfrak{B}$ is definable in $(\rho, A)$.
\end{enumerate}

The following theorem shows that the beable for the measurement
context $(\rh,A)$ may include an observable perfectly correlated 
with $A$ but excludes observables non-commuting with observables
perfectly correlated with $A$. 

\begin{theorem}\label{th:101017}\label{theorem2}\label{th:consis}
Let $(A,B)$ be a pair of commuting
self-adjoint elements in a unital C*-algebra $\cA$.
Let $\vp$ be a state of $\cA$ and let $\cB$ be a C*-subalgebra of $\cA$.
If $\vp$ is an EPR state for $(A,B)$ and  
$\cB$ is beable for $(\vp,A)$, we have
$$
\vp(|[B^{n},Z]|^2)=0
$$
for all $Z\in\cB$ and $n=1,2,\ldots$.
\end{theorem}

\begin{proof}
Let $V_{t}=\exp\{-it(A-B)\}$, where $t\in\R$.
Then, we have $[A,V_t]=0$.
Let $(\pi,\cH,\Om)$ be the GNS representation of $\cA$ induced by $\vp$.
Since $\vp$ is an EPR state for $(A,B)$, we have $\pi(A-B)\Om=0$,
and consequently $\pi(V_{t})\Om=\Om$.
Thus, we have $\vp_{V_t}=\vp$.
From (Def), we conclude $V_t^* \cB V_t=\cB$.
Now let $Z\in\cB$.
From ($A$-Priv) we have $[A,Z]\in \cB$.
Note here that $V_t^* Z V_t$ is uniformly continuous in $t$,
since $A-B$ is bounded.
Thus, we also have 
$[A-B,Z]=-i\frac{d}{dt}V_t^* Z V_t|_{t=0}\in\cB$,
so that we obtain $[B,Z]\in\cB$.
Let $\om$ be a dispersion-free state on $\cB$.
By the Hahn-Banach theorem, $\om$ can be extended to a state $\tilde{\om}$ of $\cA$.
Then, we have $\tilde{\om}(ZB)=\tilde{\om}(BZ)=\om(Z)\tilde{\om}(B)$. 
Thus, we have $\om([B,Z])=0$, so that $\om(|[B,Z]|^2)=
|\om([B,Z])|^2=0$.
Since $\vp|_{\cB}$ is classical, there is a probability distribution $\mu$ 
of dispersion-free states satisfying
$$
\vp(|[B,Z]|^2)=\int \om(|[B,Z]|^2)\,d\mu(\om)=0.
$$
This completes the proof for $n=1$.  The proof for $n>1$ is analogously
obtained by replacing $A-B$ by $A^{n}-B^{n}$.
\end{proof}

Let $\rho$ be an EPR state for incommensurable
pairs  $(A_1,A_2)$ and $(B_1,B_2)$.
Suppose that $A_1$ is measured in the state $\rh$.
Then, the measurement result on $A_1$ is considered real.
By Bohr's requirement of a classical description, $B_1$ cannot have the elements
of reality.   Since $A_2$ and $B_2$ are not commuting in $\rho$, they have no
simultaneous reality, so that the problem is on what ground to choose 
either $A_2$ or $B_2$.
According to Howard, Bohr regarded elements predictable with certainty as real 
\cite[p.256]{How79} plausibly, once given such a context.  
Therefore, $A_2$ is an element of reality if $A_1$ is measured in the state $\vp$. 
Since $A_2$ and $B_2$ cannot be ascribed simultaneous reality, 
$B_2$ is not an element of reality in this case. 
This interpretation is consistent with the objectivity requirement,
implicit definability of the beables $\cB$, or invariance under the relevant 
group of symmetries, as shown by the following theorem.

\begin{theorem}\label{th:101107b}\label{theorem3}
Let $\vp$ be an EPR state of a  unital C*-algebra 
$\cA$ for incommensurable
pairs $(A_1,A_2)$ and $(B_1,B_2)$ in $\cA$.
Then, neither $B_1$ nor $B_2$ can be an element 
of any beable subalgebra for $(\vp,A_1)$.
\end{theorem}
\begin{proof}
Let $\cB$ be a beable subalgebra for $(\vp,A_1)$.
Suppose $B_1\in\cB$. 
By \cite[Proposition 2.2]{HC99}, $\vp(|[A_1^n,B_1^m]|^2)=0$
for all $n,m$.
This contradicts the incommensurability of 
$(A_1,A_2)$ and $(B_1,B_2)$, so that $B_1\not\in\cB$.
Suppose $B_2\in\cB$. 
It follows from Theorem \ref{th:101017}
that $\vp(|[A_2^n,Z]|^2)=0$ with $Z=B_2^m\in\cB$ for all $n,m$.
This contradicts the incommensurability of 
$(A_1,A_2)$ and $(B_1,B_2)$, so that $B_2\not\in\cB$. 
\end{proof}

\section{Unbounded observables}
In the previous section, we have extended the Halvorson-Clifton 
consistency theorem to EPR states in algebraic quantum theory.  
However,  this formulation is limited to the case where
observables to be correlated by EPR states are bounded.
Since the original EPR argument involves unbounded observables
such as position and momentum, a further extension of the formulation
is necessary to include any EPR states for unbounded observables. 
This section is devoted to such an extension of the preceding
result.  In the next section, we shall show that the result obtained
in this section reconstructs Bohr's reply to the original EPR argument.

In algebraic quantum theory, any quantum system $\bS$ is represented
by a C*-algebra $\cA$.  Obviously, $\cA$ does not contain any unbounded observables.  Thus, $\cA$ ``describes'' an unbounded observable $A$ of $\bS$ 
through a class of bounded functions $f(A)$ of $A$.
Here, we assume that the C*-algebra $\cA$ describes an unbounded observable $A$ of $\bS$ if and only if $\cA$ contains the exponential
function $e^{itA}$ of $A$ for any real number $t$, or equivalently 
if and only if $\cA$ 
contains every almost periodic functions $f(A)$ of $A$.
This assumption is relevant to CCR algebras or Weyl algebras as discussed
in the next section.

Denote by ${\rm CB}(\R)$ the space of complex-valued bounded continuous functions on the real line $\R$.
We say that a function $f\in\R$ is {\em almost periodic} 
($f\in \AP(\R)$)  if its orbit under translations
$O(f)=\{f_{t}|t\in \R\}$ is relatively compact in ${\rm CB}(\R)$,
where $f_{t}(x)=f(x-t)$ for all $x\in\R$ \cite{Gre69}.
Denote by $\ch_t$ for a real $t\in\R$ the character on $\R$ 
defined by $\ch_{t}(x)=e^{itx}$ for all $x\in\R$.
The space $\AP(\R)$ is characterized as the C*-algebra generated by
the set of all characters $\ch_{t}$ with $t\in\R$,
and is *-isomorphic to the C*-algebra $C(b\R)$ of continuous
functions on the Bohr compactification $b\R$ of $\R$ \cite{Loo53}.

Let $\cA$ be a unital C*-algebra and let $\vp$ be a state of $\cA$.
An  {\em observable}  of $\cA$ is defined to be a *-homomorphism 
from the abelian C*-algebra $\AP(\R)$ of almost periodic functions on $\R$ into $\cA$.  
Denote by $\OO(\cA)$ the space of observables of $\cA$.
The range $A(\AP(\R))$ of an observable $A$ is denoted by  
$C^*(A)$, which is a unital C*-subalgebra of $\cA$.
For any observable $A\in\OO(\cA)$, we define $U^A_t\in\cA$ by
$U^A_t=A(\ch_{t})$ for $t\in\R$ and $f(A)\in\cA$ by $f(A)=A(f)$
for $f\in\AP(\R) $. 
Then, the mapping $t\mapsto U^A_t$ is a (not necessarily continuous) group
homomorphism from $\R$ to $\cU(\cA)$,
where $\cU(\cA)$ is the group of unitary elements of $\cA$.
Conversely,  for every group homomorphism  $U_t$  from $\R$ to $\cU(\cA)$,
there is a *-homomorphism $\pi$ from $\AP(\R)$ into $C^{*}(A)$
such that $\pi(f)=\sum_{k=1}^{n} \al_k U^A_{t_k}$  if 
 $f(x)=\sum_{k=1}^{n}\al_{k}e^{it_k x}$.
An observable $A$ is said to be {\em bounded},
if $t\mapsto U^A_t$ is norm continuous, i.e., 
$\lim_{t\to 0} \|U^A_t -1\|=0$.
For every bounded observable $A$, there is some $\hat{A}\in\cA_{sa}$
uniquely such that $U^A_t=e^{it\hat{A}}$.
Conversely, every $\hA\in\cA_{sa}$ has the corresponding bounded
observable $A$ so that bounded observables and self-adjoint elements
of $\cA$ are in one-to-one correspondence.
Observables
$A,B\in\OO(\cA)$ are said to be {\em commuting} if
$\CC(A)$ and $\CC(B)$ are commuting subalgebras.
Let $\vp$ be a state of $\cA$.
For any commuting $A, B\in \OO(\cA)$, we say that $\vp$ is an {\em EPR state}
for $(A,B)$ if  $\vp(U^{A}_t U^{B}_{-t})=1$ for any $t\in\R$.
This relation is equivalent to the relation
$\vp((\hA-\hB)^2)=0$, if $A,B$ are bounded,
and is equivalent to the ``perfect correlation'' between $A$ and $B$,
if  $A$ and $B$ are represented by unbounded self-adjoint operators 
and $\rh$ is represented by a density operator 
on a Hilbert space on which $\cA$ acts  \cite[Theorem 5.5]{06QPC}.

Now, we extend the notion of a measurement context to any pair $(\vp,A)$
of a state $\vp$ of $\cA$ and an unbounded observable $A$ of $\cA$.
A C*-subalgebra $\cB$ of $\cA$ is said to be {\em beable} for a measurement
context $(\vp,A)$ if
it satisfies the following conditions.
\begin{enumerate}
\item[(i)]  (Beable) $\vp$ is classical on $\cB$.

\item[(ii)] ($A$-Prev) $\CC(A)\subseteq\cB$.

\item[(iii)] (Def)  For any unitary $U\in\cA$ leaving $(\vp,A)$ invariant,
the subalgebra $\cB$ is globally invariant under $U$; namely,
 $U^{*}\cB U =\cB$, if $[U,U^A_t]=0$ for every $t\in \R$ 
and $\vp_U=\vp$.
\end{enumerate}

\begin{proposition}\label{prop:110602}
Let $A,B$ be commuting observables of $\cA$.
A state $\vp$ of $\cA$ is an EPR state for $(f(A),f(B))$ for every $f\in \AP(\R)$
if it is an EPR state for $(A,B)$.
\end{proposition}

\begin{proof}
Suppose that $\vp$ is an EPR state for  $(A,B)$.
Let $(\pi,\cH,\Om)$ be a GNS representation of $\cA$ induced by $\vp$.
By assumption,
$$
\|\pi(U^A_t)\Om-\pi(U^B_t)\Om\|^2=
2-2\Re\bracket{\pi(U^A_t)\Om,\pi(U^B_t)\Om}
=
2-2\vp(U^A_{-t}U^B_t)
=0.
$$
Thus,
$\pi(U^A_t)\Om=\pi(U^B_t)\Om$.
It follows from linearity that if 
$$
f(x)=\sum_{k=1}^{n}\al_ke^{it_k x}
$$ 
then
$\pi(f(A))\Om=\pi(f(B))\Om$, and this holds for every $f\in\AP(\R)$ by continuity.
Therefore,  we have  $\vp((f(A)-f(B))^2)=0$ for all $f\in \AP(\R)$.
\end{proof}

The following theorem extends  Theorem   \ref{th:consis}
to unbounded observables.

\begin{theorem}\label{th:101029a}\label{theorem4}
Let $\cA$ be a unital  C*-algebra and let $\vp$ be a state of $\cA$.
Let $A, B\in \OO(\cA)$.
If $\vp$ is an EPR state for  $(A,B)$ and $\cB$ is beable for $(\vp,A)$,
then we have
$$
\vp(|[f(B),Z]|^2)=0
$$
for all $Z\in\cB$ and $f\in\AP(\R)$.
\end{theorem}

\begin{proof}
The proof is obtained from the proof of Theorem \ref{theorem2}
by replacing $A-B$ by $f(A)-f(B)$ using Proposition \ref{prop:110602}
to show $\vp_{V_t}=\vp$ if $f$ is real-valued.  Then, the assertion 
follows easily. 
\end{proof}

Let $\vp$ be a state of $\cA$.
Observables $A$ and $B$ of $\cA$ are said to {\em commute in $\vp$} if
$\vp(|[U^{A}_t,U^{B}_s]|^2)=0$ for all $t,s\in\R$.
This definition is equivalent to the previous one for bounded observables
$A$ and $B$. 
Let $\cA_1,\cA_2$ be commuting subalgebras of $\cA$.
Then, pairs $(A_1,A_2)$ and $(B_1,B_2)$ with $A_j,B_j\in\cA_j$ 
for $j=1,2$ are called {\em incommensurable pairs} from $\cA_1,\cA_2$
for $\vp$ if $A_j,B_j\in\cO(\cA_j)$ for $j=1,2$ and 
$\vp(|[U^{A_j}_t,U^{B_j}_s]|^2)\not=0$ for some $t,s\in\R$.

The following theorem extends  Theorem   \ref{th:101107b}
to unbounded observables.

\begin{theorem}\label{theorem5}
Let $\cA_1,\cA_2$ be commuting subalgebras of  a unital C*-algebra $\cA$.
Let $\vp$ be an EPR state of $\cA$ 
for incommensurable pairs $(A_1,A_2)$ and $(B_1,B_2)$ from $(\cA_1,\cA_2)$. 
Then, neither $B_1$ nor $B_2$ can be an observable
of any beable subalgebra for $(\vp,A_1)$.
\end{theorem}
\begin{proof}
Let $\cB$ be a beable subalgebra for $(\vp,A_1)$.
Suppose $B_1\in\cO(\cB)$. 
By \cite[Proposition 2.2]{HC99}, 
$\vp(|[U^{A_1}_t,U^{B_1}_s]|^2)=0$ for any $t,s\in\R$.
This contradicts the incommensurability of 
$(A_1,A_2)$ and $(B_1,B_2)$, so that $B_1\not\in\cB$.
Suppose $B_2\in\cB$. 
It follows from Theorem \ref{th:101029a}
that $\vp(|[U^{A_2}_t,Z]|^2)=0$ with $Z=U^{B_2}_s\in\cB$.
This contradicts the incommensurability of 
$(A_1,A_2)$ and $(B_1,B_2)$, so that $B_2\not\in\cB$. 
\end{proof}

\section{EPR argument}
 
In this section, we formulate the original EPR argument 
following Halvorson and Clifton \cite{HC02},
and explicitly show that the general theorem obtained 
in the preceding section
can reconstruct Bohr's reply to the original EPR argument.
 
 \sloppy
Let  $\hQ, \hP$ be self-adjoint operators corresponding to 
position and momentum, respectively, in the Schr\"{o}dinger representation 
on $L^2(\R)$ of one-dimensional motion.
Let  $\cA[\R^2]$ be the Weyl algebra, namely, $\cA[\R^2]$ be the
C*-algebra generated by $f(\hQ), f(\hP)$ for all $f\in\AP(\R)$. 
We define observables $Q,P$ of $\cA[\R^2]$ by $Q(f)=f(\hQ)$ and
$P(f)=f(\hP)$ for all $f\in\AP(\R)$.
Then,  $\CC(\hQ)$ and $\CC(\hP)$ are *-isomorphic to $\AP(\R)$.
For any observable $X$ of $\cA[\R^2]$, we denote 
the corresponding observables in  $\cA[\R^2]\otimes \C I$ 
and  $\C I\otimes \cA[\R^2]$ 
by $X_1=X\otimes I$ and $X_2=I\otimes X$, respectively.

We define observables $Q_1-Q_2$ and $P_1+P_2$ of  
 $\cA[\R^2]\otimes \cA[\R^2]$ by 
$Q_1-Q_2(f)=f(\hQ_1-\hQ_2)$, $P_1+P_2(f)=f(\hP_1+\hP_2)$ 
for all $f\in \AP(\R)$. 
Denote by $C^*(Q_1-Q_2,P_1+P_2)$ the C*-algebra
generated by $C^*(Q_1-Q_2)$ and $C^*(P_1+P_2)$.
In what follows,  let
 $\cA=\cA[\R^2]\otimes \cA[\R^2]$.

\begin{proposition}
 $C^{*}(Q_1-Q_2,P_1 +P_2)$ is *-isomorphic to $\AP(\R^2)$.
\end{proposition}
\begin{proof}
Let  $\cF$ be the Fourier transform on $L^2(\R)$.
Let $\al_1(X)=(1\otimes\cF^{*})X(1\otimes\cF)$.
For every $f\in \AP(\R)$, we have
 $\al_1(f(\hQ_1))=f(\hQ_1)$ and $\al_1(f(\hQ_2))=f(\hP_2)$.
 Let  $U=e^{-i\hP_1\hQ_2/\hbar}$ and $\al_2(X)=U^{*}XU$.
Then, $\al_2$ is a *-automorphism of $\cA$.
We have $\al_2(f(\hQ_1))=f(\hQ_1-\hQ_2)$ and 
$\al_2(f(\hP_2))=f(\hP_1+\hP_2)$ for every $f\in \AP(\R)$.
Thus, we have  $\al_2\circ\al_1(f(\hQ_1,\hQ_2))=f(\hQ_1-\hQ_2,\hP_1+\hP_2)$
for all $f\in \AP(\R^2)$.
The correspondence $\iota:f(x,y)\mapsto f(\hQ_1,\hQ_2)$ is 
*-isomorphism from $\AP(\R^2)$ to $C^{*}(Q_1,Q_2)$,
and hence $C^{*}(Q_1 -Q_2,P_1 +P_2)$ is *-isomorphic to
$\AP(\R^2)$.
\end{proof}

For every $(u,v)\in\R^2$,  a {\em canonical $(u,v)$-EPR state}
is a state  $\vp$ of $\cA$ satisfying
\beql{CEPR}
 \vp(f(\hQ_1-\hQ_2,\hP_1+\hP_2))=f(u,v)
 \eeq
 for all $f\in\AP(\R)$.
 
Canonical EPR states can be constructed by the following standard argument.
By the *-isomorphism between $C^{*}(Q_1-Q_2,P_1 +P_2)$
and $\AP(\R^2)$, every point  $(u,v)\in\R^2$
defines a state  $\vp$ of $C^{*}(Q_1 -Q_2,P_1 +P_2)$ 
satisfying \Eq{CEPR}.  By the Hahn-Banach theorem, this state extends to a  
canonical  $(u,v)$-EPR state of $\cA$.
Halvorson \cite{Hal00} proved the uniqueness of the canonical $(u,v)$-EPR state
for every $u,v\in\R$.

\begin{proposition}
The canonical $(u,v)$-EPR state is an EPR state of $\cA$ for $(Q_1,u1+Q_2)$ and $(P_1,v1-P_2)$.
\end{proposition}
\begin{proof}
Let $\vp$ be the canonical $(u,v)$-EPR state.
We have
$$
\vp(U^{Q_1}_{t}U^{u1+Q_2}_{-t})=
\vp(e^{it\hQ_1}e^{-it(u1+\hQ_2)})
=e^{-itu}\vp(e^{it(\hQ_1-\hQ_{2})})=1.
$$
Thus, $\vp$ is an EPR state for $(Q_1,u1+Q_2)$, and similarly
it can be shown that it is an EPR state for $(P_1,v1-P_2)$.
\end{proof}

According to Bohr, a measurement on the first particle in an EPR state influences the condition which defines elements of reality for the second particle. The following theorem 
mathematically reconstructs Bohr's reply to EPR.

\begin{theorem}
\label{theorem9}
Let $\vp$ be the canonical  $(u,v)$-EPR state of $\cA$ for $u,v\in\R$.
Then,  neither $P_1$ nor $P_2$ can be an observable of any beable subalgebra
$\cB$ for $(\vp,Q_1)$.
\end{theorem}
\begin{proof}
Let $\cB$ be a beable subalgebra for $(\vp,Q_1)$.
Since $P_1$ does not commute with $Q_1$, it is obvious that
$P_1\not\in\cO(\cB))$.
Suppose $P_2\in\OO(\cB)$.
Since $\vp$ is an EPR state for $(Q_1,uI+Q_2)$ and $\cB$ is beable
for $(\vp,Q_1)$,  from Theorem \ref{th:101029a} with 
$Z=e^{i\pi P_2}\in\cB$ we have 
$$
\vp(|[e^{iQ_2},e^{i\pi P_2}]|^2)=0.
$$
Let $(\cH,\vpi,\Om)$ be the GNS construction of $\cA$ induced by $\vp$.  
Then, we have
$$
\vpi(e^{iQ_2}e^{i\pi P_2})\Om=\vpi(e^{i\pi P_2}e^{iQ_2})\Om.
$$
By the Weyl commutation relation, we have
$$
\vpi(e^{iQ_2}e^{i\pi P_2})\Om=-\vpi(e^{i\pi P_2}e^{iQ_2})\Om.
$$
Thus, we have
$$
2\vpi(e^{iQ_2}e^{i\pi P_2})\Om=0.
$$
This is a contradiction,
since $\vpi(e^{iQ_2}e^{i\pi P_2})$ is unitary,
so that we conclude $P_2\not\in\OO(\cB)$.
\end{proof}

\section{Concluding Remarks}
\sloppy
Let $\mathcal{O}_1$ and $\mathcal{O}_2$ be strictly space-like separated regions in Minkowski space. Then there exists 
a vector state $\vp$ of $\mathfrak{N}(\mathcal{O}_1) \vee \mathfrak{N}(\mathcal{O}_2)$ which is an EPR state for incommensurable
pairs $(E_1,E_2)$ and $(F_1,F_2)$ of projections from 
$(\mathfrak{N}(\mathcal{O}_1), \mathfrak{N}(\mathcal{O}_2))$ by Theorem \ref{theorem1}. If we were to measure $E_1$ in $\mathcal{O}_1$ in $\vp$, we could predict with certainty the outcome of $E_2$; and if we were to measure $F_1$ in $\mathcal{O}_1$ in $\psi$, we could predict with certainty the outcome of $F_2$.

If we accepted the criterion of reality which was proposed by EPR, $E_2$ and $F_2$ would
simultaneously be elements of reality.  However, quantum theory has
no theoretical counter part of them, since they do not commute. 
Therefore algebraic quantum field theory would be incomplete with this criterion.

On the other hand, Bohr proposed another criterion of reality. 
According to him, a measurement in $\mathcal{O}_1$ influences the condition which defines elements of reality in $\mathcal{O}_2$. 
Howard \cite{How94} reconstructed Bohr's reply to EPR 
in terms of Bohr's notion of  ``a classical description'' and ``a contextualized
version of the EPR reality criterion''. According to this criterion, ``our decision as to which particular properties to consider as real will turn on the question of predictability with certainty'' \cite[p.256]{How79}. Therefore $F_1$ is an element of reality if $E_1$ is measured in $\vp$. Since $F_1$ and $F_2$ cannot be ascribed simultaneous reality, $F_2$ is not an element of reality.

Halvorson and Clifton showed that Bohr's reply to EPR is consistent with the following requirements \cite[pp.16-17]{HC02}: (1) Empirical Adequacy: When an observable is measured, it possesses a determinate value distributed in accordance with the probabilities determined by the quantum state. (2) Classical Description: Properties $P$ and $P'$ can be simultaneously real in a quantum state only if that state can be represented as a joint classical probability distribution over $P$ and $P'$. (3) Objectivity: Elements of reality must be invariants of those symmetries that preserve the defining features of the measurement context.
For the EPR position-momentum case 
they proved the consistency between those two requirements
\cite[Theorem 2]{HC02},
by showing that given the measurement context $(\vp,Q_1)$ with an EPR state and the 
position $Q_1$ of the particle I, $Q_1$  has the value as an element of reality, 
the value of $P_1$ might be ascribed reality, but the value of $P_2$ cannot be
ascribed reality. 
In Theorem \ref{theorem9} we have reconstructed this consistency theorem
through an rather elementary arguments in a general formulation of algebraic quantum
theory.

It is shown that Howard's contextualized reality criterion is consistent with the above three
requirements in Theorems \ref{theorem2} and \ref{theorem3} 
for the case of bounded observables in 
algebraic quantum field theory, and in Theorems \ref{theorem4} and \ref{theorem5} 
for the case of unbounded observables.
Thus, Howard's reconstruction of Bohr's reply based on a contextualized
version of the EPR reality criterion has now acquired independent grounds 
from the objectivity requirement represented by an invariance principle
under relevant symmetries even in a general formulation of algebraic
quantum theory.

\section*{Acknowledgements}
The authors thank Jeremy Butterfield for helpful comments and suggestions
for an earlier version of this paper.
M.O. thanks Andreas Doering for his warm hospitality at the Department 
of Computer Science, Oxford University, where the final part of this work 
has been done.
M.O. is supported in part by the JSPS KAKENHI,
No.21244007 and No.22654013.
Y.K. is supported by the JSPS KAKENHI, No.23701009.

\end{document}